\documentclass[12pt,leqno,letterpaper]{article}

\usepackage{amsmath,amsthm,enumerate,amssymb}
\usepackage[utf8]{inputenc}

\usepackage[T1]{fontenc}
\usepackage[english]{babel}
\usepackage{graphicx}
\usepackage{color}

\newtheorem{theorem}{Theorem}[section]
\newtheorem{proposition}[theorem]{Proposition}

\newtheorem{definition}[theorem]{Definition}

\newcommand{\tr}{{\rm Tr\hskip -0.2em}~}

\newcommand{\nv}{\underline{\hbox{$0$}}}

\DeclareMathOperator{\frechetdiff}{\mathit d}
\newcommand{\fd}[1]{\hskip -0.2em\frechetdiff\hskip -0.32em{#1}}

\begin{document}

\title{Characterisation of matrix entropies}
\author{Frank Hansen and Zhihua Zhang}
\date{February  10, 2014\\{\small Major revision March 15, 2015}}

\maketitle

\begin{abstract}
The notion of matrix entropy was introduced by Tropp and Chen with the aim of measuring the fluctuations of random matrices.  It is a certain entropy functional constructed from a representing function with prescribed properties, and Tropp and Chen gave some examples. We give several abstract characterisations of matrix entropies together with a sufficient condition in terms of the second derivative of their representing function.\end{abstract}

\section{Introduction and main result}

The search for concentration inequalities has been a flourishing field in probability theory during the past
thirty years \cite{kn:Boucheron:2005}. Among various inequalities the matrix concentration inequality class has
applications in many fields such as random graph theory, compressed sensing et cetera \cite{kn:tropp:2014}. Recently Chen and Tropp developed a matrix extension of the entropy method and used it to search for matrix concentration inequalities.
They studied in particular the matrix entropy inequalities associated with the standard entropy function
$t \mapsto t\log t $ and the power functions $t \mapsto t^p $ for $p \in [1,2]. $ 

Let throughout this paper $ H_n $ denote a Hilbert space of finite dimension $ n, $ and let $ \mathcal H_n $ be the Hilbert space of bounded linear operators on $ H_n $  equipped with the inner product given by the trace.

Tropp and Chen \cite{kn:tropp:2014} essentially\footnote{We allow the representing function $ \varphi $ to be defined only in the positive half-line since continuity in zero is automatic for all known examples.} gave the following definition:

\begin{definition}
Let for each natural number $n$ the class $\Phi_{n}$ consist of the functions $\varphi\colon (0,\infty) \to \mathbf{R}$
that are either affine or satisfy the following three conditions.

\begin{enumerate}[(i)]

\item $\varphi$ is convex.

\item $\varphi$ is twice continuously differentiable.

\item Let $ f=\varphi' $ be the derivative of $ \varphi. $ The Fréchet differential $ \fd{f}(x) $ 
of the matrix function $ x\to f(x) $ is an invertible linear operator on the Hilbert space $ \mathcal H_n $ and the map $ x\mapsto \fd{f}(x)^{-1}$ is  concave. 

\end{enumerate}
\end{definition}

Notice that $ f $ in the above definition is an increasing function.
Condition $ (ii) $ may be omitted as it follows from the other conditions. The class of (representing functions for) matrix entropies  $\Phi_\infty $ is defined as the intersection
\[
\Phi_\infty =\bigcap _{n=1}^\infty \Phi_{n}.
\]
It follows from Theorem~\ref{Theorem: convexity of Q-form}
 that each set $ \Phi_n $ is a convex cone. This is not obvious \cite{kn:latala:2000, kn:Boucheron:2005} even for $ n=1. $  The authors  \cite{kn:Boucheron:2005} showed that 
a twice differentiable strictly convex function $\varphi$ defined in the positive half-line is in $\Phi_{1}$ if the
induced $\varphi$-entropy
\[
H_{\varphi}(Z)=\mathbb{E}[\varphi(Z)]-\varphi(\mathbb{E}[Z])
\]
is convex on the set $\mathbb{L}_{\infty}^{+}(\Omega, \mathcal{A}, \mathbb{P})$ of bounded and 
non-negative random variables $Z$, where $\mathbb{E}[Z]$ denotes the expectation of $Z$.

 More generally,
Tropp and Chen  \cite{kn:tropp:2014} introduced to each $\varphi \in \Phi_\infty$ the following matrix $\varphi$-entropy functional
\[
H_{\varphi}(Z)=\mathbb{E}[\tr\varphi(Z)]-\tr\varphi(\mathbb{E}[Z]),
\]
where now $Z$ is a positive semi-definite random matrix. The authors established subadditivity of $ H_{\varphi} $ and derived matrix extensions of the bounded difference inequality and the moment inequality by choosing suitable representing functions in $ \Phi_\infty\,. $ For the difference inequality they used the function $t \to t\log t, $ and for the moment inequality the functions $ t \to t^p, $ where $ p=q/(q-1) $ for integers $ q=2,3,\dots. $

By applying and extending the techniques in \cite{kn:hansen:2014} we are able to reformulate the defining properties of a matrix entropy in a more transparent way giving rise to several abstract characterisations as given below.

\begin{theorem}\label{characterisation of matrix entropies}
Let $ \varphi\colon (0,\infty)\to \mathbf{R} $ be a twice continuously differentiable convex function, and let $ f=\varphi' $ denote the derivative of $ \varphi. $ The following conditions are equivalent.

\begin{enumerate}[(i)]

\item $ \varphi $ is the representing function of a matrix entropy.

\item The map $ (x,h)\mapsto \tr h^*\fd{f}(x)h $ is, for each natural number $ n, $ convex in pairs of operators in $ B(H_n), $ where $ x $ is positive definite.  

\item The function of two variables
\[
(x,y)\mapsto \tr (y-x)(f(y)-f(x))
\]
is convex in positive definite operators on an arbitrary finite dimensional Hilbert space.

\item The function of two variables
\[
g(t,s)=\frac{s-t}{f(s)-f(t)}\qquad t,s>0
\]
is operator concave.

\end{enumerate}

\end{theorem}
Theorem~\ref{characterisation of matrix entropies} follows from Theorem~\ref{Theorem: convexity of Q-form}, Theorem~\ref{characterisation in terms of two variable trace function}, and Theorem~\ref{characterisation in terms of operator concavity}. 

\begin{theorem}\label{sufficient condition}

Let $ \varphi\colon (0,\infty)\to \mathbf{R} $ be a twice differentiable function, and let $ f=\varphi' $ denote the derivative of $ \varphi. $ 
If $ f' $ is positive, numerically decreasing and operator convex then:

\begin{enumerate}[(i)]

\item  $ \varphi $ is the representing function of a matrix entropy.

\item $ \varphi $ allows a continuous extension to the closed interval $ [0,\infty). $

\item\label{canonical representation of a matrix entropy} $ \varphi $ may be written in the canonical form
\[
\displaystyle \varphi(x)=a+bx+\frac{\beta}{2}x^2+\int_0^\infty (1+\lambda)\Bigl(1-x+(x+\lambda)\log\frac{x+\lambda}{1+\lambda}\Bigr)\,d\mu(\lambda),
\]
where $ a=\varphi(1)-\varphi'(1)+\beta/2 $ and $ b=\varphi'(1)-\beta, $ in terms of a uniquely defined bounded and positive measure $ \mu. $

\end{enumerate}
\end{theorem}

The uniquely defined measure $ \mu $ in the above theorem comes from the following well-known characterisation, see for example the analysis in \cite[Page 9-10]{kn:hansen:2006:1}.

\begin{proposition}\label{integral formula for positive, decreasing operator convex functions}
A positive function $ g $ defined in the positive half-line is operator convex and decreasing if and only if it can be written on the form
\[
g(t)=\beta+\int_0^\infty \frac{1+\lambda}{t+\lambda}\,d\mu(\lambda)\qquad t>0,
\]
where $ \mu $ is a positive and bounded measure and $ \beta\ge 0. $
\end{proposition}

It may be easier to establish the necessity of the integral representation by noticing that $ g $ is operator decreasing and the function $ h(t)= g(t^{-1}) $ thus operator monotone. The function
\[
h^*(t)=t h(t^{-1})=t g(t)\qquad t>0
\]
is therefore operator monotone by \cite[Corollary 4.2]{kn:hansen:2013:1}. The integral formula then follows from Corollary 5.1 in the same reference by  setting $ \beta=\mu(\{\infty\}). $ \\[1ex]
Based on massive numerical calculations we conjecture that the matrix entropies given in Theorem~\ref{sufficient condition} exhaust the class of matrix entropies.\\[1ex]
Chen and Tropp proved that the standard entropy function $ t\mapsto t\log t$ and
the power functions $t\mapsto t^{p}$ for $p\in [1,2]$ are representing functions for matrix entropies. These statements are mathematically already contained  in earlier results by Lieb \cite{kn:lieb:1973:1} and the first author \cite{kn:hansen:2014} formulated outside the theory of matrix entropies. This is an example of how different authors may arrive at similar conclusions as the result of independent research activities.

\section{Reformulating the main condition}

Theorem~\ref{characterisation of matrix entropies} follows from a number of separate results, the first being
an adaptation of a result by the first author \cite{kn:hansen:2014} applying ideas going back to Lieb \cite{kn:lieb:1973:1}. 

\begin{theorem}\label{Theorem: convexity of Q-form}

Let $ f\colon (0,\infty)\to\mathbf R $ be  a strictly increasing continuously differentiable function, and let $ n $ be a fixed natural number.  The following conditions are equivalent.

\begin{enumerate}[(i)]

\item The map $ x \mapsto \tr h^*\fd{f}(x)^{-1}h $ is, for each $ h\in\mathcal H_n, $  concave in positive definite operators $ x\in B(H_n). $

\item The map $ (x,h)\mapsto \tr h^*\fd{f}(x)h $ is convex in pairs of operators in $ B(H_n), $ where $ x $ is positive definite.  

\end{enumerate}

In condition $ (ii) $ it is sufficient to assume convexity in pairs $ (x,y), $ where $ x $ is positive definite and $ y $ is self-adjoint.  

\end{theorem}

\begin{proof}
We first assume $ (i) $ and define two quadratic forms $ \alpha $ and $ \beta $ on the direct sum $ \mathcal H_n \oplus \mathcal H_n $ by setting
\[
\begin{array}{rl}
\alpha(X\oplus Y)&=\lambda \tr X^*\fd{f}(A_1)X +(1-\lambda) \tr Y^*\fd{f}(A_2)Y\\[2ex]
\beta(X\oplus Y)&=\tr (\lambda X^*+(1-\lambda)Y^*)\fd{f}(A)(\lambda X+(1-\lambda)Y),
\end{array}
\]
where $ A_1,A_2 $ are two fixed positive definite operators in $ B(H_n), $ and $ A=\lambda A_1+(1-\lambda) A_2 $ for some $ \lambda\in[0,1]. $ The differential operator $ \fd{}f(x) $ is a super operator on $ \mathcal H_n $ defined first in self-adjoint operators by the functional calculus and then extended to $ \mathcal H_n $ by linearity.
The statement of the theorem is equivalent to the majorisation
\begin{equation}\label{quadratic form majorisation}
\beta(X\oplus Y)\le\alpha(X\oplus Y)
\end{equation}
for arbitrary $ X,Y\in \mathcal H_n\,. $
Let $ (e_i)_{i=1}^n $ be a basis in which $ x $ is diagonal and let $ \lambda_1,\dots,\lambda_n $ be the corresponding
eigenvalues counted with multiplicity. Expressed in this basis\,  $ \fd{f}(x)h=h\circ L_f\bigl(\lambda_1,\dots,\lambda_n\bigr) $ is the Hadamard (entry-wise) product of $ h $ and the L\"{o}wner matrix 
\[
L_f\bigl(\lambda_1,\dots,\lambda_n\bigr)=  \left([\lambda_i,\lambda_j]_f\right)_{i,j=1}^n ,
\]
where the divided difference $[t,s]_{f}$ is defined by setting
\[
[t,s]_{f}=\left\{
            \begin{array}{ll}\displaystyle
              \frac{f(t)-f(s)}{t-s}  \hskip 2em &t \neq s \\[2.5ex]

              f^{'}(t)                &t=s.  \\
            \end{array}
          \right.
\]The quadratic form $ h\mapsto\tr h^*\fd{f}(x)h $ is positive definite since
\[
\tr h^*\fd{f}(x)h=\sum_{i,j=1}^n |(h e_i\mid e_j)|^2[\lambda_i,\lambda_j]_{f}
\]
and $[\lambda_i,\lambda_j]_{f}>0. $ The
corresponding sesqui-linear form is given by
\[
(h,k)\to\tr k^*\fd{f}(x)h.
\]
The two quadratic forms $ \alpha $ and $ \beta $ are in particular positive definite. Therefore,
there exists an operator $ \Gamma $ on $ \mathcal H_n \oplus \mathcal H_n $ which is positive
definite in the Hilbert space structure given by $ \beta $ such that
\[
\alpha\bigl(X\oplus Y, X'\oplus Y'\bigr)=\beta\bigl(\Gamma(X\oplus Y), X'\oplus Y'\bigr)\qquad X,X',Y,Y'\in \mathcal H_n\,,
\]
where we retain the notation $ \alpha $ and $ \beta $ also for the corresponding sesqui-linear forms.
Let $ \gamma $ be an eigenvalue of $ \Gamma $ corresponding to an eigenvector $ X\oplus Y. $ Then
\[
\alpha\bigl(X\oplus Y,X'\oplus Y'\bigr)=\beta\bigl(\gamma(X\oplus Y), X'\oplus Y'\bigr)\qquad\text{for}\quad  X',Y'\in \mathcal H_n
\]
or equivalently
\[
\begin{array}{l}
\lambda\tr (X')^*\fd{f}(A_1)X+(1-\lambda)\tr (Y')^*\fd{f}(A_2)Y\\[1.5ex]
=\gamma\tr (\lambda (X')^*+(1-\lambda)(Y')^*)\fd{f}(A)(\lambda X+((1-\lambda)Y)
\end{array}
\]
for arbitrary $ X',Y'\in \mathcal H_n\,. $ We may assume $ 0<\lambda<1 $ and then derive that
\[
\fd{f}(A_1)X=\gamma\fd{f}(A)(\lambda X+(1-\lambda)Y)=\fd{f}(A_2)Y.
\]
Thus by setting $ M=\fd{f}(A)(\lambda X+(1-\lambda)Y) $ we obtain
\[
\begin{array}{rl}
\fd{f}(A)^{-1}(M)&=\lambda X+(1-\lambda) Y\\[1.5ex]
&=\lambda\fd{f}(A_1)^{-1}(\gamma M)+(1-\lambda)\fd{f}(A_2)^{-1}(\gamma M).
\end{array}
\]
By multiplying from the left with $ M^* $ and taking the trace we obtain
\[
\begin{array}{l}
\gamma\bigl(\lambda\tr M^* \fd{f}(A_1)^{-1} M + (1-\lambda)\,\tr M^* \fd{f}(A_2)^{-1}M\bigr)
=\displaystyle\tr M^* \fd{f}(A)^{-1} M\\[2.5ex]
\ge\displaystyle \lambda\tr M^* \fd{f}(A_1)^{-1} M +(1-\lambda)\tr M^* \fd{f}(A_2)^{-1} M,
\end{array}
\]
where the last inequality is implied by the concavity of $x \mapsto \tr h^*\fd{f}(x)^{-1}h .$
This shows that the operator $ \Gamma\ge 1$ from which (\ref{quadratic form majorisation})
and thus statement $ (ii) $ of the theorem follows.

If we instead assume statement $ (ii) $ in the theorem and consider the same construction as above, then the eigenvalue $ \gamma\ge 1 $ and the last inequality therefore implies that
\[
\tr M^* \fd{f}(A)^{-1} M
\ge\displaystyle \lambda\tr M^* \fd{f}(A_1)^{-1} M +(1-\lambda)\tr M^* \fd{f}(A_2)^{-1} M
\]
for each $ M\in\mathcal H_n $ on the form  $ M=d f(A)(\lambda X+(1-\lambda)Y). $ Since the Fréchet differential
$ \fd{f}(A) $  is bijective, any vector $ M\in\mathcal H_n $ may be written in this form. We conclude that the map $ x\mapsto \fd{f}(x)^{-1} $ is concave which is statement $ (i) $ in the theorem.

By replacing $ \mathcal H_n $ with the vector space of self-adjoint operators on $ H_n $ we may carry out the same construction as above without any essential changes in the proof. This shows that we may relax condition $ (ii) $ to pairs $ (x,y), $ where $ x $ is positive definite and $ y $ is self-adjoint.
\end{proof}

\section{A bivariate trace function}

Let $ \varphi\colon (0,\infty)\to \mathbf{R} $ be a differentiable function, and let $ f=\varphi' $ denote the derivative of $ \varphi. $ 

\begin{theorem}\label{characterisation in terms of two variable trace function}

Then $ \varphi\in\Phi_n $ if and only if the trace function of two variables,
\begin{equation}\label{convex function of two variables}
(x,y)\mapsto \tr (y-x)(f(y)-f(x)),
\end{equation}
is convex in positive definite $ n\times n $ matrices.
\end{theorem}

\begin{proof} 

We first assume $ \varphi\in\Phi_n $ for a fixed natural number $ n. $ Take an operator $ h $ in $ B(H_n) $ and consider arbitrary operators $ y \in B(H_n). $ By composing with the linear map $ y\mapsto yh, $ we obtain that the map 
\[
(x,y)\mapsto \tr h^* y^* \fd{f}(x)(yh)
\]
is convex in pairs of operators $ (x,y) $ where $ x $ is positive definite. Furthermore,
\[
\fd{f}(x)(yh)=\fd{f}(x)(L_y h)=(\fd{f}(x)L_y)(h),
\]
where $ L_y $ denotes left multiplication with $ y. $ The map
\[
(x,y)\to \tr h^* \bigl(L_y^* \fd{f}(x)L_y\bigr)(h)
\]
is therefore convex and since $ h $ is arbitrary, we obtain that the map
\[
(x,y)\to L_y^* \fd{f}(x) L_y\in B(\mathcal H_n)
\]
is convex in pairs of operators $ (x,y) $ where $ x $ is positive definite.  Let now also $ y $ be positive definite. To each $ t\in[0,1] $ we set  $ x_t=(1-t)x+ t y. $ By composing with the linear map  $ (x,y)\mapsto (x_t, y-x) $ we obtain that
the map
\[
(x,y)\mapsto L_{y-x} \fd{f}(x_t)L_{y-x}\in B(\mathcal H_n)
\]
is convex in pairs of positive definite $ n\times n $ matrices.  We then define an operator $ T(x,y)\in B(\mathcal H_n) $ by setting
\begin{equation}\label{two variable operator convex map}
T(x,y)= L_{y-x} \int_0^1 \fd{f}(x_t)L_{y-x}\,dt
\end{equation}
for positive definite operators $ x $ and $ y $ on $ H_n\,. $
It follows from the above that $ T $ is convex.  By taking the expectation of $ T(x,y) $ in the unit operator we obtain the identity
\begin{equation}\label{trace of T(x,y)}
\tr T(x,y)=\tr L_{y-x} \int_0^1 \fd{f}(x_t)(y-x)\,dt=\tr (y-x)(f(y)-f(x)),
\end{equation}
cf. for example \cite[Theorem 2.1]{kn:hansen:1995}. The statement then follows from the convexity of $ T. $

Suppose on the other hand that the two variable trace function defined in (\ref{convex function of two variables}) is convex in positive definite $ n\times n $ matrices. The two variable function 
\[
(x,y)\to \tr T(x,y) =\tr (y-x)\int_0^1 \fd{}f(x_t)(y-x)\,dt,
\]
where $ x_t=(1-t)x+ty, $
is then convex by the identity in (\ref{trace of T(x,y)}). In particular,
\[
\begin{array}{l}
\displaystyle\tr T\Bigl(\frac{x_1+x_2}{2}\,,\frac{(x_1+s y_1)+(x_2+s y_2)}{2}\Bigr)\\[2ex]
\hskip 8em\displaystyle\le\frac{1}{2}\bigl(\tr T(x_1, x_1+s y_1)+\tr T(x_2, x_2+s y_2)\bigr)
\end{array}
\]
for positive definite matrices $ x_1,x_2 $ and hermitian matrices $ y_1,y_2 $ and $ s>0 $ such that $ x_1+sy_1 $ and $ x_2+s y_2 $ are positive definite. This reduces to
\[
\begin{array}{l}
\displaystyle s^2\tr \frac{y_1+y_2}{2}\int_0^1 \fd{}f\Bigl(\frac{x_1+x_2}{2}+ts\frac{y_1+y_2}{2}\Bigr) \frac{y_1+y_2}{2}\,dt\\[2ex]
\displaystyle\le\frac{s^2}{2}\Bigl(\tr y_1\int_0^1 \fd{}f(x_1+ts y_1) y_1\, dt +\tr y_2\int_0^1 \fd{}f(x_2+ts y_2) y_2\, dt \Bigr).
\end{array}
\]
Since the Fréchet differential is continuous, we obtain by dividing with $ s^2 $ and then letting $ s $ tend to zero the inequality
\[
\begin{array}{l}
\displaystyle \tr \frac{y_1+y_2}{2} \fd{}f\Bigl(\frac{x_1+x_2}{2}\Bigr) \frac{y_1+y_2}{2}
\displaystyle\le\frac{1}{2}\bigl(\tr y_1\fd{}f(x_1) y_1 +\tr y_2 \fd{}f(x_2) y_2\bigr)
\end{array}
\]
showing that the map $ (x,y)\to \tr y\,\fd{}f(x)y $ is convex in pairs $ (x,y) $ of $ n\times n $ matrices, where $ x $ is positive definite and $ y $ is self-adjoint.
\end{proof}

\section{Bivariate operator convex functions}

Consider a function $ g\colon D\to\mathbf R $ of two variables defined in a convex domain $ D\subseteq\mathbf R^2. $
Let $ x $ and $ y $ be commuting self-adjoint operators on a Hilbert space of finite dimension $ n $ with spectra $ \sigma(x) $ and $ \sigma(y) $ such that $ \sigma(x)\times\sigma(y)\subset D. $ We say that $ (x,y) $ is in the domain of $ g. $ 
Consider the spectral resolutions
\[
x=\sum_{i=1}^p \lambda_i P_i\qquad\text{and}\qquad y=\sum_{j=1}^q \mu_j Q_j,
\]
where $ \lambda_1,\dots,\lambda_p $ and $ \mu_1,\dots,\mu_q $ respectively are the eigenvalues of $ x $ and $ y, $ and
$ P_1,\dots,P_p $ and $ Q_1,\dots,Q_q $ are the corresponding commuting spectral projections. The functional calculus is defined by setting
\[
g(x,y)=\sum_{i=1}^p\sum_{j=1}^q g(\lambda_i,\mu_j) P_i Q_j.
\]
\begin{definition}\label{operator convexity for bivariate functions}
The function $ g $ is said to be matrix convex of order $ n $ if for arbitrary $ * $-algebras $ \mathcal A_1 $ and $ \mathcal A_2 $ of operators acting on $ H_n $ the inequality
\begin{equation}\label{convexity condition}
g(\lambda x_1+(1-\lambda)x_2, \lambda y_1+(1-\lambda)y_2)\le\lambda g(x_1,y_1)+(1-\lambda)g(x_2,y_2)
\end{equation}
holds for $ \lambda\in[0,1] $ and operators $ x_1,y_1\in\mathcal A_1 $ and $ x_2,y_2\in\mathcal A_2 $ such that
$ (x_1,y_1) $ and $ (x_2,y_2) $ are in the domain of $ g. $
\end{definition}

Notice that under the conditions given in the above definition, the pair $ (\lambda x_1+(1-\lambda)x_2, \lambda y_1+(1-\lambda)y_2) $ is automatically in the domain of $ g. $ We say that $ g $ is operator convex if $ g $ is matrix convex of all orders.

Kor\'anyi \cite[page 542]{kn:koranyi:1961} gave a definition of the functional calculus for bivariate functions in which the  pairs to which the function is applied are mapped into commuting parts of a tensor product. This type of functional calculus is convenient in many situations, but ceases to be useful when all the relevant operators are supposed to act on the same Hilbert space. There is a certain literature exploring operator convexity for multivariate functions, where Kor\'anyi's functional calculus and its obvious generalisations to more than two variables, are applied.

It is clear that a function which is operator convex by Definition~\ref{operator convexity for bivariate functions} is also
operator convex in the sense of Kor\'anyi\footnote{We use this terminology even though Kor\'anyi did not study convex functions.}.  However, Definition~\ref{operator convexity for bivariate functions} is not an empty generalisation, and this can be gleaned from the following example: The function $ f(t,s)=(ts)^{-1} $ is operator convex with respect to Kor\'anyi's functional calculus. The result is essentially due to Lieb~\cite[Theorem 8]{kn:lieb:1973:1}, cf. also Ando~\cite[Theorem 5]{kn:ando:1979}, who proved that $ (x,y)\to x^{-1}\otimes y^{-1} $ is a convex map. The diagonal map $ x\to x^{-1}\otimes x^{-1} $ is in particular convex. In fact, Ajula~\cite[Theorem 3.1]{aujla:1993} proved that a bivariate function $ f(t,s) $ is operator convex with respect to Kor\'anyi's functional calculus, if and only if the diagonal map $ x\to f(x,x) $ is convex. One can show that the function $ f(t,s)=(ts)^{-1} $ is also operator convex according to Definition~\ref{operator convexity for bivariate functions}, but the diagonal map $ x\to f(x,x)=x^{-2} $ is not convex.

\begin{theorem}\label{characterisation in terms of operator concavity}

Let $ \varphi\colon (0,\infty)\to \mathbf{R} $ be a differentiable function, and let $ f=\varphi' $ denote the derivative of $ \varphi. $ 
Then $ \varphi $ is the representing function of a matrix entropy if and only if the bivariate function
\[
g(t,s)=\frac{s-t}{f(s)-f(t)}\qquad t,s>0
\]
is operator concave.

\end{theorem}

\begin{proof}

By calculation we obtain that the expectation of the inverse Fr\'{e}chet differential is given by
\[
\tr h^*\fd{f}(x)^{-1} h=\sum_{i,j=1}^n |(he_i\mid e_j)|^2 \frac{\lambda_i-\lambda_j}{f(\lambda_i)-f(\lambda_j)}=\tr h^*\, g(L_x,R_x) h,
\]
where $L_x$ and $R_x$ denote left and right multiplication with $x,$ respectively. If $ g $ is operator concave it follows that $ x\to\fd{}f(x)^{-1} $ is concave, thus $ \varphi $ is the representing function of a matrix entropy.

To prove the converse we consider block matrices
\[
H=\begin{pmatrix}
     0 & h\\
     0 & 0
     \end{pmatrix}\qquad\text{and}\qquad
Z=\begin{pmatrix}
           x & 0\\
           0 & y
           \end{pmatrix}.
\]     
It is a matter of simple algebra to prove the identities 
\[
\tr H^*L_Z H=\tr h^* L_x h\quad\text{and}\quad
\tr H^*R_Z H=\tr h^* R_y h.
\]
Consider two positive definite $ n\times n $ matrices $ x $ and $ y, $ and let $ (e_1,\dots,e_n) $ and $ (d_1,\dots,d_n), $ respectively, be orthonormal bases of eigenvectors of $ x $ and $ y $ such that
\[
xe_i=\lambda_i e_i\qquad\text{and}\qquad y d_i=\mu_i d_i
\]
for $ i=1,\dots,n. $ Setting 
\[
E_i=\left\{\begin{array}{ll}
              e_i\oplus \nv\qquad &i=1,\dots,n\\[1ex]
              \nv\oplus d_{i-n}      &i=n+1,\dots,2n
              \end{array}\right.
\]              
the orthonormal basis $ (E_1,\dots,E_{2n}) $ in $ H_n\oplus H_n $ diagonalises $ Z, $ and since $ \fd{}f(Z)H=L_f(Z)\circ H $ is the Hadamard product of the corresponding L{\"o}wner matrix $ L_f(Z) $ and $ H $ expressed in this basis we obtain
\[
\begin{array}{rl}
\tr H^*\fd{}f(Z)^{-1}H&=\displaystyle\sum_{i,j=1}^n |(he_i\mid d_j)|^2 \frac{\lambda_i-\mu_j}{f(\lambda_i)-f(\mu_j)}=\tr h^* g(L_x,R_y)h.
\end{array}
\]
It follows that the map
\[
(x,y)\to g(L_x,R_y)\in B(\mathcal H_n),
\]
defined in positive definite operators $ x,y\in B(H_n), $ is convex. 
Let now $ \mathcal A_1 $ and $ \mathcal A_2 $ be two commuting $ * $-algebras on a Hilbert space of finite dimension.  We want to prove that the map $ (x,y)\to g(x,y) $ is convex in positive operators  $ (x,y)\in\mathcal A_1\times\mathcal A_2. $ Since finite dimensional $ * $-algebras are direct sums of factors, it is sufficient to prove the assertion for two commuting factors $ \mathcal A_1 $ and $ \mathcal A_2 $ with $ \mathcal A_1'=\mathcal A_2. $ Such factors are in the so-called standard representation and may be represented as the left, respectively right, representation of the algebra of operators acting on some finite dimensional Hilbert space.
It follows that $ g $ is matrix convex of any order and thus operator convex.
\end{proof}

\section{Proof of Theorem~\ref{sufficient condition}}

Let $ \varphi\colon (0,\infty)\to \mathbf{R} $ be a twice continuously differentiable convex function, and let $ f=\varphi' $ denote the derivative of $ \varphi. $ Suppose that $ f' $ is operator convex and decreasing. We consider the positive function
\[
k(t,s)=\frac{f(t)-f(s)}{t-s}=\int_0^1 f'(\lambda t+(1-\lambda)s)\,d\lambda\qquad t,s>0,
\]
where we used Hermite's formula. Let now $ x $ be a positive definite operator in $ B(H_n) $ and take
an orthonormal basis $ (e_1,\dots,e_n) $ in which $ x $ is diagonal with eigenvalues given by
\[
x e_i=\lambda_i e_i\qquad i=1,\dots,n.
\]
By calculation we obtain that the expectation of the Fr\'{e}chet differential is given by
\[
\tr h^*\fd{f}(x) h=\sum_{i,j=1}^n |(he_i\mid e_j)|^2 \frac{f(\lambda_i)-f(\lambda_j)}{\lambda_i-\lambda_j}=\tr h^*\, k(L_x,R_x) h,
\]
where $L_x$ and $R_x$ denote left and right multiplication with $x,$ respectively. 
Putting the formulas together we obtain the identity 
\begin{equation}\label{transform equation between derivatives}
\tr h^*\fd{f}(x) h=\int_0^1 \tr h^* f'(\lambda L_x + (1-\lambda) R_x) h\,d\lambda.
\end{equation}
The first author proved \cite[Page 100]{kn:hansen:2006:3} that a mapping of the type
\[
(x,\xi)\mapsto (g(x)\xi\mid\xi), 
\]
where $ g\colon (0,\infty)\to\mathbf R $ is a positive function,
is convex in pairs $ (x,\xi) $ of positive definite operators $ x $ on an arbitrary Hilbert space $ H $ and vectors $ \xi\in H, $ if $ g $ is operator convex and decreasing. Subsequently, Ando and Hiai \cite[Theorem 3.1]{kn:ando:2011} proved that the condition is not only sufficient but also necessary. 

Since the transformation $ x\to \lambda L_x + (1-\lambda) R_x $ is affine and $ f' $ is operator convex and numerically decreasing, we realise that the mapping of two variables in (\ref{transform equation between derivatives}) is convex. This proves $ (i) $ in Theorem~\ref{sufficient condition}.

Notice that a function $ f\colon (0,\infty)\to\mathbf R $ with operator convex and numerically decreasing derivative $ f' $ necessarily is operator monotone. 

We now use Proposition~\ref{integral formula for positive, decreasing operator convex functions} and write
\[
f'(t)=\beta+\int_0^\infty\frac{1+\lambda}{t+\lambda}\,d\mu(\lambda)\qquad t>0
\]
for some bounded positive measure $ \mu $ and $ \beta\ge 0. $ Therefore,
\[
f(t)=f(1)+\int_1^t f'(s)\,ds=f(1)+\beta(t-1)+\int_1^t\int_0^\infty \frac{1+\lambda}{s+\lambda}\,d\mu(\lambda)\,ds.
\]
Since the positive kernel is bounded within the integration limits, and the product measure is bounded, we may apply Fubini's theorem and obtain
\[
\begin{array}{rl}
f(t)&=\displaystyle f(1)+\beta(t-1)+\int_0^\infty (1+\lambda)\int_1^t\ \frac{1}{s+\lambda}\,ds\,d\mu(\lambda)\\[2.5ex]
&=\displaystyle  f(1)+\beta(t-1)+\int_0^\infty (1+\lambda)\log\frac{t+\lambda}{1+\lambda}\,d\mu(\lambda).
\end{array}
\]
We therefore obtain the representation
\[
\varphi(x)=\varphi(1)+\int_1^x f(t)\,dt\qquad x>0
\]
and write it on the form
\[
\varphi(x)=a+bx+\frac{\beta}{2}x^2+\int_1^x \int_0^\infty (1+\lambda)\log\frac{t+\lambda}{1+\lambda}\,d\mu(\lambda)\,dt.
\]
We now consider the kernel
\[
g(t,\lambda)=(1+\lambda)\log\frac{t+\lambda}{1+\lambda}=(1+\lambda)\bigl(\log(t+\lambda)-\log(1+\lambda)\bigr).
\]
If $ x\ge 1 $ and $ 1\le t\le x $ we may use the mean value theorem and obtain
\[
0\le g(t,\lambda)=(1+\lambda)\frac{t-1}{\xi}\qquad\text{for some}\quad 1+\lambda<\xi<t+\lambda,
\]
thus $ 0\le g(t,\lambda)\le x-1 $ for $ 1\le t\le x $ and $ \lambda\ge 0. $
If on the other hand $ 0< x\le 1 $ and $ x\le t\le 1, $ then
\[
0\ge g(t,\lambda)=(1+\lambda)\frac{t-1}{\xi} \qquad\text{for some}\quad t+\lambda<\xi<1+\lambda,
\]
thus
\[
0\ge g(t,\lambda)\ge (t-1)\frac{1+\lambda}{t+\lambda}\ge (x-1)\frac{1+\lambda}{x+\lambda}\ge x-1.
\]
We have shown that the kernel $ g(t,\lambda) $ is bounded on sets $ [x,1]\times [0,\infty) $ for $ 0<x\le 1 $ and
$ [1,x]\times [0,\infty) $ for $ 1\le x. $ We may thus as above apply Fubini's theorem and obtain
\[
\varphi(x)=a+bx+\frac{\beta}{2}x^2+\int_0^\infty(1+\lambda)\int_1^x \log\frac{t+\lambda}{1+\lambda}\,dt\,d\mu(\lambda).
\]
By calculating the inner integral we obtain $ \varphi $ on the canonical form
\begin{equation}\label{canonical representation of a matrix entropy}
\displaystyle \varphi(x)=a+bx+\frac{\beta}{2}x^2+\int_0^\infty (1+\lambda)\Bigl(1-x+(x+\lambda)\log\frac{x+\lambda}{1+\lambda}\Bigr)\,d\mu(\lambda),
\end{equation}
where $ a=\varphi(1)-\varphi'(1)+\beta/2 $ and $ b=\varphi'(1)-\beta. $ This is $ (iii) $ in Theorem~\ref{sufficient condition}.
We now turn the attention to the kernel
\[
h(x,\lambda)=(1+\lambda)\Bigl(1-x+(x+\lambda)\log\frac{x+\lambda}{1+\lambda}\Bigr)
\]
when $ 0<x\le 1. $ By the mean value theorem we obtain that
\[
\log(x+\lambda)=\log(1+\lambda)+\frac{x-1}{\xi}\quad\text{and thus}\quad \log\frac{x+\lambda}{1+\lambda}=\frac{x-1}{\xi}
\]
for some $ x+\lambda<\xi<1+\lambda. $ In particular,
\[
h(x,\lambda)=(1+\lambda)(1-x)\Bigl(1-\frac{x+\lambda}{\xi}\Bigr)\qquad\lambda\ge 0.
\]
Since $ \xi>x+\lambda $ we notice that $ h(x,\lambda)\ge 0, $ and since $ \xi<1+\lambda $ we obtain
\[
h(x,\lambda)\le (1+\lambda)(1-x)\Bigl(1-\frac{x+\lambda}{1+\lambda}\Bigr)=(1-x)^2.
\]
The kernel $ h(x,\lambda) $ is therefore uniformly bounded on the set $ [0,x]\times[0,\infty) $ for $ x\le 1. $ Since the measure $ \mu $ is bounded we conclude that $ \varphi $ is continuous in zero and that $ \lim_{x\to 0}\varphi(x)\le a+\mu([0,\infty)). $
This concludes the proof of Theorem~\ref{sufficient condition}.

{\small

%\bibliographystyle{plain}
%\bibliography{/users/frankhansen/Dropbox/UNIDATA/Macros/texmf/bibtex/Bib/Mathharv}

\begin{thebibliography}{10}

\bibitem{kn:ando:1979}
T.~Ando.
\newblock Concavity of certain maps of positive definite matrices and
  applications to \uppercase{H}adamard products.
\newblock {\em Linear Algebra Appl.}, 26:203--241, 1979.

\bibitem{kn:ando:2011}
T.~Ando and F.~Hiai.
\newblock Operator log-convex functions and operator means.
\newblock {\em Mathematische Annalen}, 350(3):611--630, 2011.

\bibitem{aujla:1993}
J.S. Aujla.
\newblock Matrix convexity of functions of two variables.
\newblock {\em Linear Algebra and Its Applications}, 194:149--160, 1993.

\bibitem{kn:tropp:2014}
R.A. Chen and J.A. Tropp.
\newblock Subadditivity of matrix $ \varphi $-entropy and concentration of
  random matrices.
\newblock {\em Electron. J. Probab.}, 19(27):1--30, 2014.

\bibitem{kn:hansen:2006:3}
F.~Hansen.
\newblock Extensions of \uppercase{L}ieb's concavity theorem.
\newblock {\em Journal of Statistical Physics}, 124:87--101, 2006.

\bibitem{kn:hansen:2006:1}
F.~Hansen.
\newblock Trace functions as \uppercase{L}aplace transforms.
\newblock {\em Journal of Mathematical Physics}, 47:043504, 2006.

\bibitem{kn:hansen:2013:1}
F.~Hansen.
\newblock The fast track to \uppercase{L}{\"o}wner's theorem.
\newblock {\em Linear Algebra Appl.}, 438:4557--4571, 2013.

\bibitem{kn:hansen:2014}
F.~Hansen.
\newblock Trace functions with applications in quantum physics.
\newblock {\em J. Stat. Phys.}, 154:807--818, 2014.

\bibitem{kn:hansen:1995}
F.~Hansen and G.K. Pedersen.
\newblock Perturbation formulas for traces on $ \uppercase{C}^* $-algebras.
\newblock {\em Publ. RIMS, Kyoto Univ.}, 31:169--178, 1995.

\bibitem{kn:koranyi:1961}
A.~Kor{\'a}nyi.
\newblock On some classes of analytic functions of several variables.
\newblock {\em Trans. Amer. Math. Soc.}, 101:520--554, 1961.

\bibitem{kn:latala:2000}
R.~Latala and C.~Oleszkiewick.
\newblock Between \uppercase{S}obolev and \uppercase{P}oincar{\'e}.
\newblock volume 1745 of {\em Lecture Notes in Mathematics}, chapter Geometric
  Aspects of Functional Analysis, Israel Seminar (GAFA), pages 147--168.
  Springer, Berlin, 1996-2000.

\bibitem{kn:lieb:1973:1}
E.~Lieb.
\newblock Convex trace functions and the
  \uppercase{W}igner-\uppercase{Y}anase-\uppercase{D}yson conjecture.
\newblock {\em Advances in Math.}, 11:267--288, 1973.

\bibitem{kn:Boucheron:2005}
G.~Lugosi S.~Boucheron, O.~Bousquet and P.~Massart.
\newblock Moment inequalities for functions of independent random variables.
\newblock {\em The Annals of Probability}, 33(2):514--560, 2005.

\end{thebibliography}

\noindent Frank Hansen: Institute for Excellence in Higher Education, Tohoku University, Japan.
Email: frank.hansen@m.tohoku.ac.jp.\\[1ex]
\noindent Zhihua Zhang: School of Mathematical Sciences, University of Electronic Science and Technology of China, P. R. China, and Department of Mathematics, Graduate School of Science, Tohoku University, Japan.
}

\end{document}